\documentclass[a4paper,UKenglish]{lipics-vhugo}

\usepackage{hyperref}
\usepackage{amsfonts,amsmath,wrapfig}

\usepackage{latexsym}
\usepackage{xspace}
\usepackage{rotating}
\usepackage{amsmath,amssymb}
\usepackage{mathtools}
\usepackage{dsfont}
\usepackage{enumerate}
\usepackage{cite}
\usepackage{microtype}

\usepackage{graphicx}
\usepackage{wasysym}

\usepackage{tikz}
\usetikzlibrary{positioning,shapes,petri}
\usetikzlibrary{automata}
\usetikzlibrary{arrows,backgrounds,positioning,fit,calc}
\usetikzlibrary{matrix}

\usepackage{enumitem}
\usepackage{macros} 

\newcommand{\hugo}{\color{black}}
\newcommand{\blaise}{\color{black}}

\newcommand{\arkiv}[1]{#1}
\newcommand{\notarkiv}[1]{}

\title{Controlling a Population}
\author[1]{Nathalie Bertrand}
\author[2]{Miheer Dewaskar}
\author[3]{Blaise Genest}
\author[4]{Hugo Gimbert}

\affil[1]{Inria, IRISA, Rennes, France}
\affil[2]{CMI, Chennai, India}
\affil[3]{CNRS, IRISA, Rennes, France}
\affil[4]{CNRS, LaBRI, Bordeaux, France}

\authorrunning{N. Bertrand, M. Dewaskar, B. Genest and H. Gimbert} 
\Copyright{N. Bertrand, M. Dewaskar, B. Genest and H. Gimbert}

\begin{document}

\maketitle
\begin{abstract}
  We introduce a new setting where a population of agents, each
  modelled by a finite-state system, are controlled uniformly: the
  controller applies the same action to every agent. The framework is
  largely inspired by the control of a biological system, namely a
  population of yeasts, where the controller may only change the
  environment common to all cells. We study a synchronisation problem
  for such populations: no matter how individual agents react to the
  actions of the controller, the controller aims at driving all agents
  synchronously to a target state. The agents are naturally
  represented by a non-deterministic finite state automaton (NFA), the
  same for every agent, and the whole system is encoded as a 2-player
  game. The first player (\playerone) chooses actions, and the second
  player (\playertwo) resolves non-determinism for each agent. The
  game with $m$ agents is called the $m$-population game. This gives
  rise to a parameterized control problem (where control refers to 2
  player games), namely the \emph{population control problem}: can
  \playerone\ control the $m$-population game for all
  $m \in \mathbb{N}$ whatever \playertwo\ does?

  In this paper, we prove that the population control
  problem is decidable, and it is a \EXPTIME-complete problem. 
  As far as we know, this is one of the first results on
  parameterized control. 
  Our algorithm, not based on cut-off techniques, produces winning strategies which are symbolic, that is, they do not need to count precisely how the population is spread between states. We also show that if there is no winning strategy, then there is a population size $\cutoff$ such that \playerone\ wins the
  $m$-population game if and only if $m\leq \cutoff$.  Surprisingly,
  $\cutoff$ can be doubly exponential in the number of states of the
  NFA, with tight upper and lower bounds.  
\end{abstract}

\section{Introduction}
Finite-state controllers, implemented by software, find applications in many different domains: telecommunication, planes, etc. There have been many theoretical studies from the model checking community to show that finite-state controllers are sufficient to control systems in idealised settings. Usually, the problem would be modeled as a game: some players model the controller, and some players model the system~\cite{ArnoldWalukiewicz}, the game settings (number of players, their power, their observation) depending on the context.

Lately, finite-state controllers have been used to control living organisms, such as a population of yeasts~\cite{Batt}. In this application, microscopy is used to monitor the fluorescence level of a population of yeasts, reflecting the concentration of some molecule, which differs from cell to cell. Finite-state systems can model a discretisation of the population of yeasts~\cite{Batt,AGKV16}.  The frequency and duration of injections of a sorbitol solution can be controlled, being injected uniformly into a solution in which the yeast population is immerged. However, the response of each cell to the osmotic stress induced by sorbitol varies, influencing the concentration of the fluorescent molecule.
The objective is to control the population to drive it through a  sequence of predetermined fluorescence states.

In this paper, we model this system of yeasts in an {\em idealised} setting: we require the (perfectly-informed) controller to surely lead synchronously all agents of a population to a state (one of the predetermined fluorescence states). Such a population control problem does not fit in traditional frameworks from the model checking community.  We thus introduce the \emph{$m$-population game}, where a population of $m$ identical agents is controlled uniformly. Each agent is modeled as a nondeterministic finite-state automaton (NFA), the same for each agent. The first player, called \playerone, applies the same action, a letter from the NFA alphabet, to every agent. Its opponent, called \playertwo, chooses the reaction of each individual agent. These reactions can be different due to non determinism.
The objective for \playerone\ is to gather all agents synchronously in the target state (which can be a sink state w.l.o.g.), and \playertwo\ seeks the opposite objective.
While this idealised setting may not be entirely satisfactory,
it constitutes a simple setting, as a first step towards more complex settings.

Dealing with large populations {\em explicitly} is in general intractable due to the state-space explosion problem. We thus consider the associated {\em symbolic parameterized control problem}, asking to reach the goal independently of the population size. 
We prove that this problem is decidable. While {\em parameterized verification} received recently quite some attention (see related work), our results are one of the first on {\em parameterized control}, as far as we know.


\medskip {\bf Our results.}  We first show that considering an infinite population is not equivalent to the parameterized control problem for all non zero integer $m$: there are cases where \playerone\ cannot control an infinite population but can control every finite population.  Solving the $\infty$-population game reduces to checking a reachability property on the support graph \cite{Martyugin-tocs14}, which can be easily done in \PSPACE.  On the other hand, solving the parameterized control problem requires new proof techniques, data structures and algorithms.

We easily obtain that when the answer to the population control problem is negative, there exists a population size $\cutoff$, called the \emph{cut-off}, such that \playerone\ wins the $m$-population game if and only if $m\leq \cutoff$. Surprisingly, we obtain a lower-bound on the cut-off doubly exponential in the number of states of the NFA. Following usual cut-off techniques would thus 
yield an inefficient algorithm of complexity at least 
2\EXPTIME.

To obtain better complexity, we developped new proof techniques ({\em not} based on cut-off techniques). Using them, we prove that the population control problem is \EXPTIME-complete. As a byproduct, we obtain a doubly exponential upper-bound for the cut-off, matching the lower-bound.
Our techniques are based on a reduction to a parity game with exponentially many states and a polynomial number of priorities.  The parity game gives some insight on the winning strategies of \playerone\ in the $m$-population games.  \playerone\ selects actions based on a set of {\em transfer graphs}, giving for each current state the set of states at time $i$ from which agent came from, for different values of $i$.  We show that it suffices for \playerone\ to remember at most a quadratic number of such transfer graphs, corresponding to a quadratic number of indices $i$. If \playerone\ wins this parity game then he can uniformly apply his winning strategy to all $m$-population games, just keeping track of these transfer graphs,
independently of the exact count in each state. If \playertwo\ wins the parity game then he also has a uniform winning strategy in $m$-population games, for $m$ large enough, which consists in splitting the agents evenly among all transitions of the transfer graphs.

\medskip {\bf Related work.}  Parameterized verification of systems with many identical components started with the seminal work of German and Sistla in the early nineties~\cite{GS-jacm92}, and received recently quite some attention.  The decidability and complexity of these problems typically depend on the communication means, and on whether the system contains a leader (following a different template) as exposed in the recent survey~\cite{Esparza-stacs14}.  This framework has been extended to timed automata templates~\cite{AJ-tcs03,ADRST-formats11} and probabilistic systems with Markov decision processes templates~\cite{BF-fsttcs13,BFS-fossacs14}.  Another line of work considers population protocols~\cite{AADFP-podc04,EGLM-concur15}. Close in spirit, are broadcast protocols~\cite{esparza-verif-99}, in which one action may move an arbitrary number of agents from one state to another. Our model can be modeled as a subclass of broadcast protocols, where broadcasts emissions are self loops at a unique state, and no other synchronisation allowed. The parameterized reachability question considered for broadcast protocols is trivial in our framework, while our parameterized control question would be undecidable for broadcast protocols.  In these different works, components interact directly, while in our work, the interaction is indirect via the common action of the controller. Further, the problems considered in related work are pure verification questions, and do not tackle the difficult issue of synthesising a controller for all instances of a parameterized system, which we do.

There are very few contributions pertaining to parameterized games
with more than one player.
The most related is \cite{Kouv}, which proves
decidability of control of mutual exclusion-like protocols
in the presence of an unbounded number of agents.
Another contribution in that domain is the one of broadcast networks of identical parity games~\cite{BFS-fossacs14}. However, the game is used to solve a verification (reachability) question rather than a parametrized control problem as in our case. Also the roles of the two players are quite different.

The winning condition we are considering is close to {\em synchronising words}. The original synchronising word problem asks for the existence of a word $w$ and a state $q$ of a {\em deterministic} finite state automaton, such that no matter the initial state $s$, reading $w$ from $s$ would lead to state $q$ (see \cite{Volkov-lata08} for a survey). Lately, synchronising words have been extended to NFAs~\cite{Martyugin-tocs14}. Compared to our settings, the author
assumes a possibly infinite population of agents, who could leak  arbitrarily often from a state to another. The setting is thus not parametrized, and a usual support arena suffices to obtain a \PSPACE algorithm. Synchronisation for probabilistic models~\cite{DMS12,DMS14}
have also been considered: the population of agents is not finite nor discrete, but rather continuous, represented as a distribution.
The distribution evolves deterministically with the choice of the controller (the probability mass is split according to the probabilities of the transitions), while in our setting, each agent can non deterministically move. This continuous model makes the parameterized verification question moot. In \cite{DMS12}, the controller needs to apply the same action whatever the state the agents are in (like our setting), and then the existence of a controller is undecidable. In \cite{DMS14}, the controller can choose the action depending on the state each agent is in (unlike our setting), and the existence of a controller reaching uniformly a set of states is \PSPACE-complete.

Last, our parameterized control problem can be encoded as a 2-player game on VASS~\cite{BJK-icalp10}, with one counter per state of the NFA: the opponent gets to choose the population size (encoded as a counter value), and for each action chosen by the controller, the opponent chooses how to move each agent (decrementing a counter and incrementing another). However, such a reduction yields a symmetrical game on VASS in which both players are allowed to modify the counter values, in order to check that the other player did not cheat. Symmetrical games on VASS are undecidable~\cite{BJK-icalp10}, and their asymmetric variant (in which only one player is allowed to change the counter values) are decidable in 2\EXPTIME~\cite{JLS-icalp15}, thus with higher complexity than our specific parameterized control problem.


\section{The population control problem}
\label{sec:setting}
\subsection{The $m$-population game}
A nondeterministic finite automaton (NFA for short) is a tuple
$\nfa=(\states, \Sigma, \state_0,\Delta)$ with $\states$ a finite set
of states, $\Sigma$ a finite alphabet, $\state_0 \in Q$ an initial state, and $\Delta \subseteq Q\times \Sigma \times Q$ the transition
relation.  
We assume throughout the paper that NFAs are complete, that
is, $\forall \state \in \states, a \in \Sigma \;, \exists p \in Q:
\;(\state,a,p) \in \Delta$. In the following, incomplete NFAs,
especially in figures, have to be understood as completed with a sink
state.

For every integer $m$, we consider a system $\nfa^{m}$ with $m$
identical agents $\nfa_1, \ldots, \nfa_{m}$ of the NFA $\nfa$.
The system $\nfa^{m}$ is itself an NFA $(\states^m, \Sigma,
\state^{m}_0,\Delta^m)$ defined as follows.  
Formally, states of
$\nfa^{m}$ are called configurations, and they are tuples $\vstate =
(\state_1,\ldots,\state_{m})\in Q^{m}$ describing the current state of each agent in the population.  We use the shorthand $\vstateinit[m]$,
or simply $\vstateinit$ when $m$ is clear from context, to denote the
initial configuration $(\stateinit,\ldots,\stateinit)$ of $\nfa^{m}$.
Given a target state $\targetstate \in Q$, the
$\targetstate$-synchronizing configuration is $\targetstate^{m} =
(\targetstate,\ldots,\targetstate)$ in which each agent is in the target state.

The intuitive semantics of $\nfa^m$ is that at each step, the same
action from $\Sigma$ applies to all agents. The effect of the action
however may not be uniform given the nondeterminism present in $\nfa$:
we have $((\state_1,\ldots,\state_m),a,(\state'_1,\ldots,\state'_m))
\in \Delta^m$ iff $(\state_j,\action,\state'_j)\in \Delta$ for all
$j \leq m$. A (finite or infinite) play in $\nfa^m$ is an alternating
sequence of configurations and actions, starting in the initial
configuration: $\play =\vstateinit \action_0 \vstate_1 \action_1
\cdots$ such that $(\vstate_i,\action_i,\vstate_{i+1}) \in
\Delta^m$ for all $i$. 

This is the \emph{$m$-population game} between \playerone\ and
\playertwo, where \playerone\ chooses the actions and \playertwo\
chooses how to resolve non-determinism. The objective for \playerone\
is to gather all agents synchronously in $\targetstate$ while
\playertwo\ seeks the opposite objective.

Our parameterized control problem asks whether \playerone\ can win the
$m$-population game for every $m \in \nats$. A strategy of \playerone\
in the $m$-population game is a function mapping finite plays to
actions, $\strat: (\states^m \times \actions)^* \times \states^m \to
\Sigma$.  A play $\play =\vstateinit \action_0 \vstate_1 \action_1
\vstate_2\cdots$ is said to \emph{respect $\sigma$}, or is a
\emph{play under $\strat$}, if it satisfies $\action_i = \strat(
\vstateinit \action_0 \vstate_1 \cdots \vstate_i)$ for all $i\in
\nats$. A play $\play =\vstateinit \action_0 \vstate_1 \action_1
\vstate_2\cdots$ is \emph{winning} if it hits the
$\targetstate$-synchronizing configuration, that is $\vstate_j =
\targetstate^m$ for some $j\in \nats$.  \playerone\ wins the
$m$-population game if he has a strategy such that all plays under
this strategy are winning. One can assume without loss of generality that $f$ is a sink state. If not, it suffices to add a new action leading tokens from $f$ to the new target sink state $\smiley$
and tokens from other states to a losing sink state $\frownie$.
The goal of this paper is to study the following parameterized
control problem:

\vspace{-.3cm}
 \begin{fmpage}{0.99\textwidth}
\textsf{Population control problem}\\
  {\bf Input}: An NFA $\nfa = (\states,\stateinit,\state_u,\actions,\trans)$
  and a target state $\targetstate \in \states$.\\
  {\bf Output}: Yes iff for every integer $m$ \playerone\ wins the $m$-population game. \end{fmpage}

	
		
		
		
		

\begin{figure}[t!]
\centering
\begin{tikzpicture}
\draw(-2,0) node [circle,draw,inner sep=2pt,minimum
size=12pt] (s1) {$\stateinit$};

\draw(0,1) node [circle,draw,inner sep=2pt,minimum size=12pt] (s2)
{$\state_1$};

\draw(0,-1) node [circle,draw,inner sep=2pt,minimum size=12pt] (s3)
{$\state_2$};

\draw(1.5,0) node [circle,draw,inner sep=2pt,minimum size=12pt] (s4) {$\targetstate$};

\draw [-latex'] (s1) -- (s2) node [pos=.5,below] {$\delta$};
\draw [-latex'] (s1) -- (s3) node [pos=.5,above] {$\delta$};

\draw [-latex']  (s2) .. controls +(60:30pt) and +(120:30pt) .. (s2)
node[midway,above]{$\delta$};
\draw [-latex']  (s3) .. controls +(240:30pt) and +(300:30pt) .. (s3)
node[midway,below]{$\delta$};

\draw [-latex'] (s2) .. controls +(160:1cm) and +(60:1cm)  .. (s1) node [pos=.5,above] {$b$};
\draw [-latex'] (s2) -- (s4) node [pos=.5,above] {\quad \, $a$};
\draw [-latex'] (s3) -- (s4) node [pos=.5,below] {\quad \, $b$};
\draw [-latex'] (s3) .. controls +(200:1cm) and +(300:1cm)  .. (s1) node [pos=.5,below] {$a$};

 \draw [-latex']  (s1) .. controls +(150:30pt) and +(210:30pt) .. (s1)
 node[midway,left]{$a,b$};
\draw [-latex']  (s4) .. controls +(30:30pt) and +(330:30pt) .. (s4)
node[midway,right]{$a,b,\delta$};
\end{tikzpicture}
	\caption{An exemple of NFA: The splitting gadget $\splitnfa$.}
		\label{fig:splitgadget}	
\end{figure}

For a fixed $m$, the winner of the $m$-population game can be
determined by solving the underlying reachability game with $|Q|^m$
states, which is intractable for large values of $m$.  
On the other hand, the answer to
the population control problem gives the winner of the $m$-population
game for arbitrary large values of $m$.  To obtain a decision
procedure for this 
parameterised problem, 
new data structures and algorithmic tools 
need to be developed, much more elaborate 
than the standard algorithm solving reachability games.

\begin{example}
\label{ex}
We illustrate the population control problem with the example
$\splitnfa$ on alphabet $\{a,b,\delta\}$ in Figure~\ref{fig:splitgadget}. 
Here, to represent
configurations we use a counting abstraction, and identify $\vstate$
with the vector $(n_0,n_1,n_2,n_3)$, where $n_0$ is the number of
agents in state $\state_0$, and so on. Under these notations, there is
a way to gather agents synchronously to $\targetstate$. We can give a
symbolic representation of a memoryless winning strategy $\strat$:
$\forall k_0,k_1 >0,\ \forall k_2,k_3 \geq 0,\ \strat(k_0,0,0,k_3) =
\delta,\ \strat(0,k_1,k_2,k_3) = a,\ \strat(0,0,k_2,k_3)=b$. Indeed,
the number of agents outside $\targetstate$ decreases by at least one
at every other step.
  The properties of this example will be detailed later and play a part in proving a lower bound (see Proposition~\ref{prop:cutoff-lowerbound}).
\end{example}

\subsection{Parameterized control and cut-off}
\label{sec:cutoff}
A first observation for the population control problem is that
$\vstateinit[m]$, $\targetstate^m$ and $Q^m$ are stable under a
permutation of coordinates.  A consequence is that the $m$-population
game is also symmetric under permutation, and thus the set of winning
configurations is symmetric and the winning strategy can be chosen
uniformly from symmetric winning configurations. Therefore, if
\playerone\ wins the $m$-population game then he has a positional
winning strategy which only counts the number of agents in each state
of $\nfa$ (the counting abstraction used in Example~\ref{ex}).


\begin{proposition}
\label{prop:monotony}
  Let $m \in \nats$. 
  If \playerone\ wins the $m$-population game, then he wins the $m'$-population game for every $m' \leq m$.
 \end{proposition}
\arkiv
{
\begin{proof}
  Let $m \in \nats$, and assume $\strat$ is a winning strategy for \playerone\ in $\nfa^m$. For $m' \leq m$ we define $\strat'$ as a strategy on $\nfa^{m'}$, inductively on the length of finite plays. Initially, $\strat'$ chooses the same first action as $\strat$: $\strat'(\stateinit^{m'}) = \strat(\stateinit^m)$. We then arbitrarily choose that the missing $m-m'$ 
agents would behave similarly as the first agent. This is indeed a possible move for the adversary in $\nfa^m$. Then, for any finite play under $\strat'$ in $\nfa^{m'}$, say $\play' = \vstateinit^{m'} \action_0 \vstate_1^{m'} \action_1 \vstate_2^{m'}\cdots \vstate_n^{m'}$, there must exist an extension $\play$ of $\play'$ obtained by adding $m-m'$ 
agents, all behaving as the first agent in $\nfa^{m'}$, that is consistent with $\strat$. Then, we let $\strat'(\play') = \strat(\play)$. Obviously, since $\strat$ is winning in $\nfa^m$, $\strat'$ is also winning in $\nfa^{m'}$.
\end{proof}
}
 The idea to define $\strat_{m'}$ is to simulate the missing $m{-}m'$
 agents arbitrarily and apply $\strat_{m}$.

  \medskip

 Hence, when the answer to the population control problem is negative,
 there exists a \emph{cut-off}, that is a value $\cutoff \in \nats$
 such that for every $m < \cutoff$, \playerone\ has a winning strategy
 in $\nfa^m$, and for every $m \geq \cutoff$, he has no winning
 strategy.

\begin{example}
  To illustrate the notion of cut-off, consider the NFA 
  on alphabet $A \cup \{b\}$ from
  Figure~\ref{fig:linear_cutoff}. Unspecified transitions lead to a
  sink state $\frownie$.

  \medskip

  The cut-off is $\cutoff = |Q|-2$ in this case. Indeed, we have the following two directions:

\medskip

  On the one hand, for $m < \cutoff$, there is a winning strategy
$\strat_m$ in $\nfa^m$ to reach $\targetstate^m$, in just two
steps. It first plays $b$, and because $m<\cutoff$, in the next
configuration, there is at least one state $q_i$ such that no agent is
in $q_i$. It then suffices to play $a_i$ to win.

  \medskip

Now, if $m \geq \cutoff$, there is no winning strategy to synchronize
in $\targetstate$, since after the first $b$, agents can be spread so
that there is at least one agent in each state $q_i$. From there, 
\playerone\ can either play action $b$ and restart the whole game, or
play any action $a_i$, leading at least one agent to the  sink state $\frownie$.
\end{example}

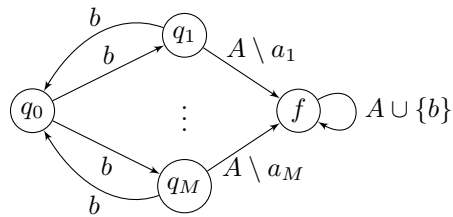
\begin{figure}[b!]
\centering
\begin{tikzpicture}
\draw(-2,0) node [circle,draw,inner sep=2pt,minimum
size=12pt] (s1) {$\stateinit$};

\draw(0,1) node [circle,draw,inner sep=2pt,minimum size=12pt] (s2)
{$\state_1$};
\draw(0,0) node [inner sep=2pt,minimum size=12pt] (s23) {$\vdots$};

\draw(0,-1) node [circle,draw,inner sep=2pt,minimum size=12pt] (s3)
{$\state_\cutoff$};

\draw(1.5,0) node [circle,draw,inner sep=2pt,minimum size=12pt] (s4) {$\targetstate$};

\draw [-latex'] (s1) -- (s2) node [pos=.5,above] {$b$};
\draw [-latex'] (s1) -- (s3) node [pos=.5,below] {$b$};

\draw [-latex'] (s2) .. controls +(160:1cm) and +(60:1cm)  .. (s1) node [pos=.5,above] {$b$};
\draw [-latex'] (s2) -- (s4) node [pos=.5,above] {\quad \, $A \setminus a_1$};
\draw [-latex'] (s3) -- (s4) node [pos=.5,below] {\quad \, $A \setminus a_\cutoff$};
\draw [-latex'] (s3) .. controls +(200:1cm) and +(300:1cm)  .. (s1) node [pos=.5,below] {$b$};

\draw [-latex']  (s4) .. controls +(30:30pt) and +(330:30pt) .. (s4)
node[midway,right]{$A \cup \{b\}$};
\end{tikzpicture}
	\caption{Illustration of the cut-off.}
		\label{fig:linear_cutoff}	
\end{figure}

\newpage

\subsection{Main results}

Our main result is 
the decidability 
of the population control problem:

\begin{theorem}
\label{th1}
  The population control problem is \EXPTIME-complete.
\end{theorem}

When the answer to the population control problem is positive, there
exists a symbolic strategy $\sigma$, applicable to all instances
$\nfa^m$, that does not need to count the number of agents in each
state. This symbolic strategy requires exponential memory.
Otherwise, 
the cut-off is at most doubly exponential, which is asymptotically tight.

\begin{theorem}
\label{th2}
In case the answer to the population control problem is negative, the
cut-off is at most $\leq 2^{2^{O(|Q|^4)}}$. There is a family of
NFA $(\nfa_n)$ of size $O(n)$ and whose cut-off is $2^{2^{n}}$.
\end{theorem}

\section{The \ror\ game}
\label{sec:capacity-game}
The objective of this section is to show that the population control
problem is equivalent to solving a game called the \emph{\ror\ game}.
To introduce useful notations, we first recall the population game
with infinitely many agents, as studied in \cite{Martyugin-tocs14} (see also \cite{Shi14} p.81).

\subsection{The $\infty$-population game}
\label{sec:resolution}
\label{subsec:game-support}
To study the $\infty$-population game, the behaviour of infinitely
many agents is abstracted into \emph{supports} which keep track of the
set of states in which at least one agent is. We thus introduce the
\emph{support game}, which relies on the notion of \emph{transfer
  graphs}. Formally, a transfer graph is a subset of $Q\times Q$
describing how agents are moved during one step.  The domain of a
transfer graph $G$ is $\dom(G) = \{ \state \in Q \mid \exists
(\state,r) \in G\}$ and its image is $\im(G) = \{r \in Q \mid \exists
(\state,r) \in G\}$.  Given an NFA $\nfa =(\states, \Sigma,
\state_0,\Delta)$ and $a \in \Sigma$, the transfer graph $G$ is
compatible with $a$ if for every edge $(q,r)$ of $G$, $(q,a,r) \in
\Delta$. We write $\GG$ for the set of transfer graphs.

The \emph{support game} of an NFA $\nfa$ is a two-player reachability
game played by \playerone\ and \playertwo\ on the \emph{support arena}
as follows. States are supports, \emph{i.e.}, non-empty subsets of
$\states$ and the play starts in $\{\state_0\}$. 
The goal support is $\{\targetstate\}$.
From a support $S$,
first \playerone\ chooses a letter $a \in \Sigma$, then \playertwo\
chooses a transfer graph $G$ compatible with $a$ and such that
$\dom(G)=S$, and the next support is $\im(G)$.  A play in the support
arena is described by the sequence $\rho = S_0 \arr{a_1,G_1} S_1
\arr{a_2,G_2}\ldots$ of supports and actions (letters and transfer
graphs) of the players.
Here, \playertwo\ best strategy is to play the maximal graph possible (this is not the case with discrete populations), and we obtain a \PSPACE-complete algorithm~\cite{Martyugin-tocs14}:
\begin{proposition}
\label{prop:support-implies-suresync}
\playerone\ wins the $\infty$-population game 
iff he wins the support game.
\end{proposition}

{\hugo 
As a consequence, the winner of the $\infty$-population game  can be computed in \EXPTIME.
However, this is of no use for deciding the population control problem,
because \playerone\ might win every $m$-population game (with
$m <\infty$) and at the same time lose
the $\infty$-population game.
This is demonstrated by the example from Figure~\ref{fig:splitgadget}.
As already shown, \playerone\ wins any $m$-population game with
$m <\infty$. However, \playertwo\ can win the $\infty$-population game
by splitting agents from $q_0$ to both
$q_1$ and $q_2$ each time \playerone\ plays $\delta$.
This way, the sequence of supports is
$\{q_0\} \{q_1,q_2\} (\{q_0,f\} \{q_1,q_2,f\})^*$, which never hits
$\{f\}$.}

\subsection{Realisable plays}

  
  {\hugo
  Plays of the $m$-population game (for $m
<\infty$) can be abstracted as plays in the support game,
by forgetting the identity of agents and keeping only track 
of edges that are used by at least one agent.
Formally, 
given a play 
 $\play =\vstateinit \action_0 \vstate_1 \action_1
\vstate_2\cdots$
of the $m$-population game,
define for every integer $n$, $S_n = \{\state \in Q \mid
\exists 1\leq i\leq m, \vstate_m[i]=\state\}$
and $G_{n+1}=\{ (s,t) \mid \exists 1\leq i\leq m, \vstate_n[i]=s \land \vstate_{n+1}[i]=t\}$.
Then  $S_0 \arr{\action_1,G_1} S_1
\arr{\action_2,G_2}\ldots$ is a play in the support arena, denoted $\Phi_m(\play)$
and called the projection of $\play$.

Not every play in the support arena can be obtained by projection, as 
in the example from Figure~\ref{fig:splitgadget} where
some plays of the support game use infinitely often the edge from 
$q_1$ to $f$, however in any $m$-population game each of the $m$ agents
might use this edge at most once.
We distinguish between these two types of plays using the notion of realisable plays.
}

\begin{definition}[Realisable plays]
A play of the support game is \emph{realisable} if there exists $m
<\infty$ such that it is the projection by $\Phi_m$ of a play
in the $m$-population game.
\end{definition}

{\hugo
A key observation is that realisability can be characterized in terms of capacity.
}

\begin{definition}[Plays with finite and bounded capacity]
  Let $\rho = S_0 \arr{a_1,G_1} S_1 \arr{a_2,G_2}\ldots$ be a play in
  the support arena.  
  
  An \emph{accumulator} of $\rho$ is a sequence
  $T=(T_j)_{j\in \NN}$ such that for every integer $j$, $T_j \subseteq
  S_j$, and which is \emph{successor-closed} \emph{i.e.}, for every $j
  \in \nats$,
{\hugo$
(s \in T_j \land  (s,t)\in G_{j+1} )\implies t \in T_{j+1}\enspace.
$}
For every $j \in \nats$, an edge $(s,t)\in G_{j+1}$ is an \emph{entry} to
$T$ if $s\not \in T_j$ and $t\in T_{j+1}$.

{\hugo
A play has \emph{finite capacity} if all its accumulators have finitely many entries,
\emph{infinite capacity} otherwise, and  \emph{bounded capacity} if
the number of entries of its accumulators is bounded.
}
\end{definition}
{\hugo
Bounded capacity is actually equivalent to realisability.
\begin{lemma}\label{lem:noleak}
A play is realisable iff it has bounded capacity.
\end{lemma}
}

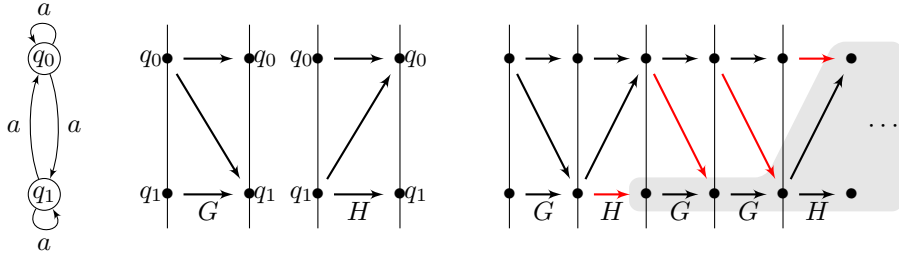
\begin{figure}
\begin{center}
\begin{tikzpicture}[scale=0.9]
\draw (-1.8,0) node [circle,draw,inner sep=0pt,minimum size=12pt ]
  (s0) {$\state_0$} ;

  \draw(-1.8,-2) node [circle,draw,inner sep=0pt,minimum size=12pt ]
  (s1) {$\state_1$} ;

 \draw [-latex'] (s1) .. controls +(110:20pt) and +(250:20pt)  .. (s0)
 node [pos=.5,left] {$a$};
 \draw [-latex'] (s0) .. controls +(290:20pt) and +(70:20pt)  .. (s1)
 node [pos=.5,right] {$a$};

\draw [-latex']  (s0) .. controls +(60:20pt) and +(120:20pt) .. (s0)
node[midway,above]{$a$};

\draw [-latex']  (s1) .. controls +(240:20pt) and +(300:20pt) .. (s1)
node[midway,below]{$a$};


\draw (0,.5) -- (0,-2.5);
\draw (1.2,.5) -- (1.2,-2.5);

\draw (0,0) node (q0Gs) {$\bullet$};
\node  at (q0Gs.180) {$\state_0$}; 
\draw (0,-2) node (q1Gs) {$\bullet$};
\node  at (q1Gs.180) {$\state_1$}; 
\draw(1.2,0) node (q0Gt) {$\bullet$};
\node  at (q0Gt.0) {$\state_0$}; 
\draw(1.2,-2) node (q1Gt) {$\bullet$};
\node  at (q1Gt.0) {$\state_1$}; 

\draw[-latex',thick] (q0Gs) -- (q0Gt);
\draw[-latex',thick] (q0Gs) -- (q1Gt);
\draw[-latex',thick] (q1Gs) -- (q1Gt);

\node at (.6, -2.25) (G) {$G$};


\draw (2.2,.5) -- (2.2,-2.5);
\draw (3.4,.5) -- (3.4,-2.5);

\draw (2.2,0) node (q0Hs) {$\bullet$};
\node  at (q0Hs.180) {$\state_0$}; 
\draw (2.2,-2) node (q1Hs) {$\bullet$};
\node  at (q1Hs.180) {$\state_1$}; 
\draw(3.4,0) node (q0Ht) {$\bullet$};
\node  at (q0Ht.0) {$\state_0$}; 
\draw(3.4,-2) node (q1Ht) {$\bullet$};
\node  at (q1Ht.0) {$\state_1$};

\draw[-latex',thick] (q1Hs) -- (q0Ht);
\draw[-latex',thick] (q1Hs) -- (q1Ht);
\draw[-latex',thick] (q0Hs) -- (q0Ht);

\node at (2.8, -2.25) (H) {$H$};


\fill [fill=black!10,rounded corners] (6.75,-2.25) -- (10.75,-2.25)  -- (10.75,.25) --
(9.75,.25) -- (8.75,-1.75) -- (6.75,-1.75) -- cycle;

\draw (5,.5) -- (5,-2.5);
\draw (6,.5) -- (6,-2.5);

\node at (5.5,-2.25) (G1) {$G$};

\draw (5,0) node (q0Gs) {$\bullet$};
\draw (5,-2) node (q1Gs) {$\bullet$};
\draw(6,0) node (q0Gt) {$\bullet$};
\draw(6,-2) node (q1Gt) {$\bullet$};

\draw[-latex',thick] (q0Gs) -- (q0Gt);
\draw[-latex',thick] (q0Gs) -- (q1Gt);
\draw[-latex',thick] (q1Gs) -- (q1Gt);

\draw (7,.5) -- (7,-2.5);

\draw (7,0) node (q0Hs) {$\bullet$};
\draw (7,-2) node (q1Hs) {$\bullet$};

\draw[-latex',thick] (q1Gt) -- (q0Hs);
\draw[-latex',thick,red] (q1Gt) -- (q1Hs);
\draw[-latex',thick] (q0Gt) -- (q0Hs);

\node at (6.5,-2.25) (H1) {$H$};
\draw (8,.5) -- (8,-2.5);
\draw (9,.5) -- (9,-2.5);
\draw (8,0) node (q0Gs) {$\bullet$};
\draw (8,-2) node (q1Gs) {$\bullet$};
\draw(9,0) node (q0Gi) {$\bullet$};
\draw(9,-2) node (q1Gi) {$\bullet$};

\draw[-latex',thick] (q0Hs) -- (q0Gs);
\draw[-latex',thick,red] (q0Hs) -- (q1Gs);
\draw[-latex',thick] (q1Hs) -- (q1Gs);
\draw[-latex',thick] (q0Gs) -- (q0Gi);
\draw[-latex',thick,red] (q0Gs) -- (q1Gi);
\draw[-latex',thick] (q1Gs) -- (q1Gi);

\node at (7.5,-2.25) (G2) {$G$};
\node at (8.5,-2.25) (G3) {$G$};

\draw(10,0) node (q0Hs) {$\bullet$};
\draw(10,-2) node (q1Hs) {$\bullet$};

\draw[-latex',thick,red] (q0Gi) -- (q0Hs);
\draw[-latex',thick] (q1Gi) -- (q0Hs);
\draw[-latex',thick] (q1Gi) -- (q1Hs);

\node at (9.5,-2.25) (G4) {$H$};

\node at (10.5,-1) (dots) {$\cdots$};
\end{tikzpicture}
\end{center}
\caption{An NFA, two transfer graphs, and a play with finite yet
  unbounded capacity.}
\label{fig:capacity}
\end{figure}

An example is given on Figure~\ref{fig:capacity} which represents an NFA, two transfer graphs $G$
and $H$, and a play $G H G^2 H G^3 \cdots$.  Obviously, this play is not realisable because at least $n$
agents are needed to realise $n$ transfer graphs $G$ in a row: at each
$G$ step, at least one agent moves from $q_0$ to $q_1$, and no new
agent enters $q_0$.  A simple analysis shows that there are only two
kinds of non-trivial accumulators $(T_j)_{j \in \nats}$ depending on
whether their first non-empty $T_j$ is $\{\state_0\}$ or
$\{\state_1\}$. We call these top and bottom accumulators,
respectively. All accumulators have finitely many entries, thus the
play has finite capacity.  However, for every $n \in \nats$ there is a
bottom accumulator with $2 n$ entries. As an example, a bottom
accumulator with $4$ entries (in red) is depicted on the figure.  Therefore, the capacity of this play is not bounded.

\arkiv{
\begin{proof}[Proof of Lemma~\ref{lem:noleak}]
  Let $\rho = S_0\arr{a_1,G_1} S_1 \arr{a_2,G_2}\cdots$ be a
  realisable play in the support arena and
  $\play=\vstate_0\vstate_1\vstate_2 \cdots$ a play in the
  $m$-population game for some $m$, such that $\Phi_m(\play) =
  \rho$. For any accumulator $T=(T_j)_{j\in \NN}$ accumulator of
  $\rho$, let us show that $T$ has less than $m$ entries.  For every
  $j\in\NN$, we define $n_j=\mid \{ 1 \leq k \leq m \mid \vstate_j(k)
  \in T_j\}\mid $ as the number of agents in the accumulator at index
  $j$.  By definition of the projection, every edge $(s,t)$ in $G_j$
  corresponds to the move of at least one agent from state $s$ in
  $\vstate_j$ to state $t$ in $\vstate_{j{+}1}$.  Thus, since the
  accumulator is successor-closed, the sequence $(n_j)_{j\in\NN}$ is
  non-decreasing and it increases at each index $j$ where the
  accumulator has an entry. The number of entries is thus bounded by
  $m$ the number of agents.

  Conversely, assume that a play $\rho = S_0\arr{a_1,G_1} S_1
  \arr{a_2,G_2}\cdots$ has bounded capacity, and let $m$ be an upper
  bound on the number of entries of its accumulators. Let us show that
  $\rho$ is the projection of a play
  $\play=\vstate_0\vstate_1\vstate_2 \cdots$ in the
  $(|S_0||Q|^{m+1})$-population game. In the initial configuration
  $\vstate_0$, every state in $S_0$ contains $|Q|^{m+1}$ agents.
  Then, configuration $\vstate_{n+1}$ is obtained from $\vstate_{n}$
  by splitting evenly the agents among all edges of $G_{n+1}$. As a
  consequence, for every edge $(s,t)\in G_{n+1}$ at least a fraction
  $\frac{1}{|Q|}$ of the agents in state $s$ in $\vstate_{n}$ moves to
  state $t$ in $\vstate_{n+1}$. By induction,
  $\play=\vstate_0\vstate_1\vstate_2 \cdots$ projects to some play
  $\rho' = S'_0\arr{a_1,G'_1} S'_1 \arr{a_2,G'_2}\cdots$ such that for
  every $n\in \NN$, $S'_n\subseteq S_n$ and $G'_n\subseteq G_n$. To
  prove that $\rho'=\rho$, we show that for every $n\in \NN$ and state
  $t\in S_n$, at least $|Q|$ agents are in state $t$ in $\vstate_{n}$.
  For that let $(U_j)_{j\in 0\ldots n}$ be the sequence of subsets of
  $Q$ defined by $U_n= \{t\}$, and for $0 < j < n$,
  \[
U_{j-1} = \{ s \in Q \mid \exists t' \in U_j, (s,t')\in G_j\}\enspace.
  \]
  Let $(T_j)_{j\in \NN}$ be the sequence of subsets of states defined
  by $T_j=Q\setminus U_j$ if $j\leq n$ and $T_j=Q$ otherwise.  Then
  $(T_j)_{j\in \NN}$ is an accumulator: if $s\not \in U_j$ and
  $(s,s')\in G_j$ then $s'\not \in U_{j{+}1}$.  As a consequence,
  $(T_j)_{j\in \NN}$ has at most $m$ entries, thus there are at most
  $m$ indices $j\in \{0\ldots n-1\}$ such that some agents in the
  states of $U_{j}$ in configuration $\vstate_{j}$ may move to states
  outside of $U_{j{+}1}$ in configuration $\vstate_{j{+}1}$. In other
  words, if we denote $M_j$ the number of agents in the states of
  $U_{j}$ in configuration $\vstate_{j}$ then there are at most $m$
  indices where the sequence $(M_j)_{j\in 0\ldots n}$ decreases.  By
  definition of $\play$, even when $M_j > M_{j{+}1}$ at least a fraction
  $\frac{1}{|Q|}$ of the agents moves from $U_{j}$ to $U_{j{+}1}$ along
  the edges of $G_{j{+}1}$, thus $M_{j{+}1} \geq
  \frac{M_j}{|Q|}$. Finally, the number of agents $M_n$ in state $t$
  in $\vstate_n$ satisfies $M_n \geq \frac{|S_0||Q|^{m+1}}{|Q|^m} \geq
  |Q|$.  Hence $\rho$ and $\rho'$ coincide, so that $\rho$ is
  realisable.
\end{proof}
}


{\hugo

\subsection{The \ror\ game}
An obvious hint to obtain a game on the support arena equivalent with
the population control problem is to make the winning condition
tougher for \playertwo, letting him lose whenever the play is not
realisable, \emph{i.e.} whenever the play has unbounded capacity.
{\hugo 
However, bounded capacity is not a regular property for runs, so that
it is not convenient to use it as a winning condition.
On the contrary, \emph{finite capacity} is a regular property,
and the corresponding 
abstraction of the population game, called the \ror\ game,
can be used to decide the population control problem.
}

\begin{definition}[\Ror\ game]
  The \emph{\ror\ game} is the game played on the support arena, where
  \playerone\ wins a play iff either the play reaches
  $\{\targetstate\}$ or the play has infinite capacity.
  A player \emph{wins the \ror\ game} if he has a winning strategy in this game.
\end{definition}

\begin{theorem}
The answer to the population control problem is positive iff
\playerone\ wins the \ror\ game, which is decidable.
\end{theorem}

This theorem is a direct corollary of the following proposition:

\begin{proposition}
\label{th.2wins}
Either \playerone\  
or \playertwo\ wins the capacity game.
The winner has a winning strategy with
\emph{ finite memory $\MemSet$}
of size $\mathcal{O}\left(2^{2^{|Q|}}\right)$.
In case \playerone\ is the winner of the \ror\ game,
he wins all $m$-population games,
for every integer $m$.
In case \playertwo\ is the winner of the \ror\ game,
he  wins the ${|\stnb|}^{ 1+ |\MemSet|
  \cdot 4^{|\stnb|}}$-population game.
\end{proposition}

\notarkiv{\begin{proof}[Proof of first and second assertions]}
\arkiv{\begin{proof}}
We start with the first assertion.
Whether a play has {infinite} capacity
can be verified by a non-deterministic B\"uchi automaton of size
$2^{|Q|}$ on the alphabet of transfer graphs, which guesses an
accumulator on the fly and checks that it has infinitely many entries.
This B\"uchi automaton can be determinised into a parity automaton
(\emph{e.g.}\ using Safra's construction) with state space $\MemSet$
of size $\mathcal{O}\left(2^{2^{|Q|}}\right)$.  The synchronized
product of this deterministic parity automaton with the support game
produces a parity game which is equivalent with the capacity game, in
the sense that, up to unambigous synchronization with the
deterministic automaton, plays and strategies in both games are the
same and the synchronization preserves winning plays and strategies.
Since parity games are determined and positional~\cite{zielonka},
either \playerone\ or \playertwo\ has a positional winning strategy in
the parity game, thus either \playerone\ or \playertwo\ has a winning
strategy with finite memory $\MemSet$ in the capacity game.
Let us prove the second assertion.
Assuming that \playerone\ wins the \ror\ game with a strategy $\sigma$,
he can win any
$m$-population game, $m <\infty$, with the strategy $\sigma_m=\sigma\circ \Phi_m$.
By definition,
the projection of every infinite play of the $m$-population game
is realisable thus has bounded capacity (Lemma~\ref{lem:noleak}).
This holds in particular for every play consistent with $\sigma_m$.
Since the projection of such a play is consistent with $\sigma$,
and since $\sigma$ wins the \ror\ game
then any  play under $\sigma_m$ reaches $\{\targetstate\}$.
%
\arkiv{
%
%

\medskip

We now prove the third and last assertion:
 if \playertwo\ has a winning strategy \emph{with finite memory
    $\MemSet$} in the \ror\ game, he has a winning strategy in the
  ${|\stnb|}^{ 1+ |\MemSet| \cdot 4^{|\stnb|}}$-population game.
  Let $\tau$ be a winning strategy for \playertwo\ in the \ror\ game
  with finite-memory $\MemSet$. We define $m={|\stnb|}^{1 +
    |\MemSet| \cdot 4^{|\stnb|}}$ and consider the $m$-population
  game.

  A winning strategy $\tau_m$ for \playertwo\ in the $m$-population
  game can be designed using $\tau$ as follows.  When it is
  \playertwo's turn to play in the $m$-population game, the play so
  far $\play=\vstate_0\arr{a_1}\vstate_1 \cdots \vstate_n
  \arr{a_{n+1}}$ is projected via $\Phi_m$ to a play $\rho =
  S_0\arr{a_1,G_1} S_1 \cdots S_n \arr{a_{n+1}}$ in the capacity game.
  Let $G_{n+1}=\tau(\rho)$ be the decision of \playertwo\ at this
  point in the capacity game.  Then, to determine $\vstate_{n+1}$,
  $\tau_m$ splits evenly the agents in $\vstate_n$ along every edge of
  $G_{n+1}$.  This guarantees that for every edge $(q,r)\in G_{n+1}$,
  at least a fraction $\frac{1}{|\stnb|}$ of the agents in state $q$
  in $\vstate_n$ moves to state $r$ in $\vstate_{n+1}$. Assuming that
  $\tau_m$ is properly defined, then it is winning for
  \playertwo. Indeed, $\tau$ guarantees that $\{f\}$ is never reached
  in the capacity game, thus $\tau_m$ guarantees that not all agents
  are simultaneously in the target state $f$.

  Now, strategy $\tau_m$ is properly defined as long as the projection
  $\rho$ is consistent with $\tau$, which in turns holds as long as at
  least one agent actually moves along every edge of $G_{n+1}$.  To
  establish that $\tau_m$ is well-defined, it is enough to show that:
\begin{itemize}
\item[($\dagger$)] for every $n\in\NN$ and every state $r \in S_n$, at
  least $|Q|$ agents are in state $r$ in $\vstate_n$ \enspace.
\end{itemize}
To show $(\dagger)$, we consider $\rho=S_0\arr{a_1,G_1} S_1
\arr{a_2,G_2}\ldots S_n$ the projection in the support arena of a play
$\play=\vstate_0\arr{a_1}\vstate_1\arr{a_2} \ldots \vstate_n$
consistent with $\tau_m$.  Let $(U_j)_{j\in 0\ldots n}$ be the
sequence of subsets of $Q$ defined by $U_n= \{r\}$, and for $0 < j <
n$,
  \[
U_{j-1} = \{ s \in Q \mid \exists t \in U_j, (s,t)\in G_j\}\enspace.
  \]
  Let $T = (T_j)_{j\in \NN}$ be the sequence of complement subsets:
  $T_j=Q\setminus U_j$ if $j\leq n$ and $T_j=Q$ otherwise.  Then, $T$
  is an accumulator: if $s\not \in U_j$ and $(s,s')\in G_j$ then
  $s'\not \in U_{j{+}1}$.
  
  Assume that there are two integers $0\leq i < j \leq n$
  such that at step $i$ and $j$
  \begin{itemize}
\item the memory state of $\tau$ coincide: $\mathsf{m}_i=\mathsf{m}_j$;
\item the supports coincide: $S_i=S_j$; and
\item the supports in the accumulator $T$ coincide: $T_i=T_j$.
\end{itemize}
Then we show that there is no entry in the accumulator between indices
$i$ and $j$.  The play $\play_*$ identical to $\play$ up to date $i$
and which repeats ad infinitum the subplay of $\play$ between dates
$i$ and $j$, is consistent with $\tau$, because
$\mathsf{m}_i=\mathsf{m}_j$ and $S_i=S_j$.  The corresponding sequence
of transfer graphs is $G_0,\ldots , G_{i-1} (G_i, \ldots, G_{j-1})
^\omega$ and $T_0,\ldots ,T_{i-1}(T_i\ldots T_{j-1})^\omega$ is a
"periodic" accumulator of $\play_*$.  By periodicity, this accumulator
has either no entry, or infinitely many entries after date $i-1$.
Since $\tau$ is winning, $\play_*$ has finite capacity, thus the
periodic accumulator has no entry after date $i-1$, and there is no
entry in the accumulator $(T_j)_{j\in \NN}$ between indices $i$ and
$j$.

Let $I$ be the set of indices where there is an entry in the
accumulator $(T_j)_{j\in \NN}$.  According to the above, for all pairs
of distinct indices $(i,j)$ in $I$, we have $m_i\neq m_j\lor S_i\neq
S_j \lor V_i\neq V_j$.  As a consequence,
\[
 |I| \leq |\MemSet|\cdot 4^{|\stnb|}\enspace.
\]
Denote $a_i$ the number of agents in $U_i$ at date $i$.  If $i\not \in
I$, \emph{i.e.}\ if there is no entry to $T_i$ at date $i$ then all
agents in $U_i$ at date $i$ are in $U_{i+1}$ at date $i+1$ hence
$a_{i+1}=a_i$. In the other case, when $i\in I$, strategy $\tau_m$
sends at least a fraction $\frac{1}{|\stnb|}$ of the agents from $U_i$
to $U_{i+1}$ thus $a_{i+1}\geq \frac{a_i}{|\stnb|}$.  Finally
\[
a_n\geq \frac{m}{|Q|^{|I|}} \geq m \cdot |\stnb|^{-|\MemSet|\cdot 4^{|\stnb|}}
=
{|\stnb|}^{1 + |\MemSet| \cdot 4^{|\stnb|}}\cdot |\stnb|^{-|\MemSet|\cdot 4^{|\stnb|}}
=|Q|\enspace.
\]
Since $U_n=\{r\}$ then property ($\dagger$) holds.  As a consequence
$\tau_m$ is well-defined and, as already discussed, $\tau_m$ is a
winning strategy for \playertwo\ in the $m$-population game.
}
\end{proof}

As consequence of Proposition~\ref{th.2wins}, the population control
problem can be decided by explicitely computing the parity game and
solving it, thus in 2\EXPTIME.
This complexity bound can  actually be improved to \EXPTIME, as shown in the next section.

  \begin{figure}
\centering
\begin{tikzpicture}[scale=0.7]
\draw (-3.5,.75) node [circle,draw,inner sep=2pt,minimum size=12pt ]
  (s0) {$\state_0$} ;
\draw (3.5,.75) node [circle,draw,inner sep=2pt,minimum size=12pt ]
  (sf) {$\targetstate$} ;

  \draw(0,3) node [circle,draw,inner sep=2pt,minimum size=12pt ]
  (s1) {$\state_1$} ;

\draw(0,1.5) node [circle,draw,inner sep=2pt,minimum size=12pt] (s2) {$\state_2$}
;

\draw(0,0) node [circle,draw,inner sep=2pt,minimum size=12pt] (s3) {$\state_3$};

\draw(0,-1.5) node [circle,draw,inner sep=2pt,minimum size=12pt] (s4){$\state_4$};

 \draw [-latex'] (s2) .. controls +(110:20pt) and +(250:20pt)  .. (s1)
 node [pos=.5,left] {$a$};
 \draw [-latex'] (s1) .. controls +(290:20pt) and +(70:20pt)  .. (s2)
 node [pos=.5,right] {$a$};

 \draw [-latex'] (s4) .. controls +(110:20pt) and +(250:20pt)  .. (s3)
 node [pos=.5,left] {$a$};
 \draw [-latex'] (s3) .. controls +(290:20pt) and +(70:20pt)  .. (s4)
 node [pos=.5,right] {$a$};

 \draw [-latex'] (s4) .. controls +(150:30pt) and +(210:30pt)  .. (s3)
 node [pos=.5,left] {$b$};
 \draw [-latex'] (s3) .. controls +(330:30pt) and +(30:30pt)  .. (s4)
 node [pos=.5,right] {$b$};

 \draw [-latex'] (s3) .. controls +(150:30pt) and +(210:30pt)  .. (s2)
 node [pos=.5,left] {$b$};
 \draw [-latex'] (s2) .. controls +(330:30pt) and +(30:30pt)  .. (s3)
 node [pos=.5,right] {$b$};

\draw [-latex']  (s1) .. controls +(60:30pt) and +(120:30pt) .. (s1)
node[midway,above]{$b$};

\draw [-latex'] (s0) .. controls +(60:2cm) and +(180:2cm)  .. (s1)
 node [pos=.5,above] {$c$};
\draw [-latex'] (s0) .. controls +(30:2cm) and +(180:2cm)  .. (s2)
 node [pos=.5,above] {$c$};
\draw [-latex'] (s0) .. controls +(330:2cm) and +(180:2cm)  .. (s3)
 node [pos=.5,below] {$c$};
\draw [-latex'] (s0) .. controls +(300:2cm) and +(180:2cm)  .. (s4)
 node [pos=.5,below] {$c$};

\draw [-latex'] (s1) .. controls +(0:2cm) and +(120:2cm)  .. (sf)
 node [pos=.5,above right] {$c$};
\draw [-latex'] (s3) .. controls +(0:2cm) and +(210:2cm)  .. (sf)
 node [pos=.5,below] {$c$};
\draw [-latex'] (s4) .. controls +(0:2cm) and +(240:2cm)  .. (sf)
 node [pos=.5,below right] {$c$};
\end{tikzpicture}
\caption{Population game where \playerone\ needs memory to win the associated \ror\ game.}
\label{fig:leak}
\end{figure}
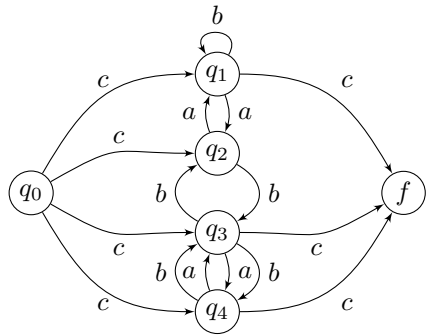

\medskip

We conclude with an example showing  that, in general, positional strategies are not sufficient to win the \ror\ game.}
%
Consider the example of Figure~\ref{fig:leak},
where the only way for \playerone\ to win is to reach a support
without $\state_2$ and play $c$.
With a memoryless strategy,
\playerone\ cannot win the capacity game.
There are only two memoryless strategies from support
$S=\{\state_1,\state_2,\state_3,\state_4\}$.
If \playerone\ only plays $a$ from $S$, the support remains
$S$ and the play has bounded capacity. 
If he only plays $b$'s from $S$, then \playertwo\ can split tokens from
$\state_3$ to both $\state_2,\state_4$ and the play remains in
support $S$, with bounded capacity. In both cases, the play 
has finite capacity
and \playerone\ loses. 

However, \playerone\ can win the capacity game. His (finite-memory)
winning strategy $\strat$ consists in first playing $c$, and then
playing alternatively $a$ and $b$, until the support does not contain
$\{\state_2\}$, in which case he plays $c$ to win.  Two consecutive
steps $ab$ send $\state_2$ to $\state_1$, $\state_1$ to $\state_3$,
$\state_3$ to $\state_3$, and $\state_4$ to either $\state_4$ or
$\state_2$.  To prevent \playerone\ from playing $c$ and win,
\playertwo\ needs to spread from $\state_4$ to both $\state_4$ and
$\state_2$ every time $ab$ is played. Consider the accumulator $T$
defined by $T_{2i}=\{\state_1,\state_2,\state_3\}$ and
$T_{2i-1}=\{\state_1,\state_2,\state_4\}$ for every $i >0$. It has an
infinite number of entries (from $\state_4$ to $T_{2i}$).  Hence
\playerone\ wins if this play is executed.  Else, \playertwo\
eventually keeps all agents from $\state_4$ in $\state_4$ when $ab$ is
played, implying the next support does not contain
$\state_2$. Strategy $\strat$ is thus a winning strategy for
\playerone.

\section{Solving the \ror\ game in \EXPTIME}
\label{sec:parity}
To solve efficiently the \ror\ game, we build an equivalent
exponential size parity game with a polynomial number of parities.  To
do so, we enrich the support arena with a \emph{tracking list}
responsible of checking whether the play has finite capacity.  The
tracking list is a list of transfer graphs, 
which are used to detect certain patterns called \emph{leaks}.

\subsection{Leaking graphs}
In order to detect whether a play $\rho = S_0 \arr{a_1,G_1} S_1
\arr{a_2,G_2}\ldots$ has finite capacity, it is enough to detect
\emph{leaking} graphs (characterising entries of accumulators). Further, leaking graphs have special
\emph{separation} properties which will allow us to 
track a small number of graphs. For $G,H$ two graphs, we denote $(a,b) \in G \cdot H$ iff there exists $z$ with $(a,z) \in G,$ and $(z,b) \in H$.

\begin{definition}[Leaks and separations]
  Let $G,H$ be two transfer graphs.  We say that $G$ \emph{leaks at
    $H$} if there exist states $q,x,y$ with $(q,y) \in G \cdot H$,
  $(x,y) \in H$ and $(q,x) \notin G$.  We say that $G$
  \emph{separates} a pair of states $(r,t)$ if there exists $q \in
  \states$ with $(q,r)\in G$ and $(q,t)\not\in G$.
\end{definition}
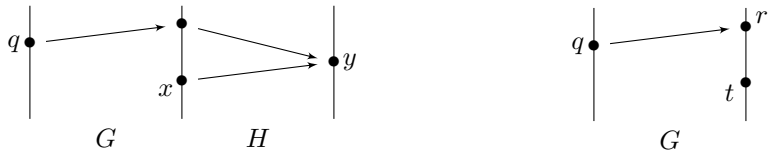
\begin{figure}[htbp]
\begin{subfigure}{.5\textwidth}
\begin{center}
\begin{tikzpicture}
\draw[thin] (0,.5) -- (0,-1);
\draw (2,.5) -- (2,-1);
\draw (4,.5) -- (4,-1);
\draw (0,0) node (q) {$\bullet$};
\node  at (q.180) {$q$}; 
\draw(2,-.5) node (x) {$\bullet$};
\node  at (x.210) {$x$}; 
\draw(2,.25) node (x') {$\bullet$};
\draw(4,-.25) node (y) {$\bullet$};
\node  at (y.0) {$y$}; 
\draw[-latex'] (q) -- (x');
\draw[-latex'] (x') -- (y);
\draw[-latex'] (x) -- (y);

\node at (1, -1.25) (G) {$G$};
\node at (3, -1.25) (H) {$H$};

\end{tikzpicture}
\end{center}
\end{subfigure}
\begin{subfigure}{.4\textwidth}
\begin{center}
\begin{tikzpicture}
\draw[thin] (0,.5) -- (0,-1);
\draw (2,.5) -- (2,-1);
\draw (0,0) node (q) {$\bullet$};
\node  at (q.180) {$q$}; 
\draw(2,-.5) node (x) {$\bullet$};
\node  at (x.210) {$t$}; 
\draw(2,.25) node (x') {$\bullet$};
\node  at (x'.30)  {$r$}; 
\draw[-latex'] (q) -- (x');
\node at (1, -1.25) (G) {$G$};
\end{tikzpicture}
\end{center}
\end{subfigure}
\caption{Left: $G$ leaks at $H$; Right: $G$ separates $(r,t)$.}
\end{figure}

The tracking list will be composed of concatenated graphs {\em tracking i} of the form $G[i,j]= G_{i+1} \cdots G_j$ relating $S_i$ with $S_j$:
$(s_i,s_j) \in G[i,j]$ if there exists $(s_k)_{i < k < j}$
with $(s_k,s_{k+1}) \in G_{k+1}$ for all $i \leq k\leq j$. 
Infinite capacity relates to leaks in the following way:
\begin{lemma}
\label{lemma.leaks}
A play has infinite capacity iff there exists an index $i$ 
such that $G[i,j]$ leaks at $G_{j+1}$ for infinitely many indices $j$.
\end{lemma}
\arkiv{
\begin{proof}
  To prove the right-to-left implication, assume that there exists an
  index $i$ such that $G[i,j]$ leaks at $G_{j{+}1}$ for an infinite
  number of indices $j$.  As the number of states is finite, there
  exist a state $q$ with an infinite number of indices $j$ such that
  we have some $(x_j,y_{j+1}) \in G_{j{+}1}$ with $(q,y_{j+1}) \in G[i,j{+}1]$, $(q,x_j) \notin G[i,j]$.  The accumulator generated by $T_i = \{q\}$ has an infinite number of entries, and we are done with this direction.

\begin{center}
\begin{tikzpicture}
\node at (0,.75) {$i$};
\node at (3,.75) {$j$};
\node at (5,.75) {$j{+}1$};

\draw (0,.5) -- (0,-1);
\draw (3,.5) -- (3,-1);
\draw (5,.5) -- (5,-1);

\draw (0,-.5) node (q) {$\bullet$};
\node  at (q.180) {$q$}; 

\draw(3,-.5) node (x) {$\bullet$};
\draw(3,0) node (x') {$\bullet$};
\node  at (x'.210) {$x$}; 
\draw(5,-.25) node (y) {$\bullet$};
\node  at (y.0) {$y$}; 
\draw[-latex',thick] (q) -- (x);
\draw[-latex',thick] (x') -- (y);
\draw[-latex',thick] (x) -- (y);

\node at (4,-1) {$G_{j{+}1}$};
\node at (1.5,-1) {$G[i,j]$};
\end{tikzpicture}
\end{center}

\medskip

For the left-to-right implication, assume that there is an accumulator
$(T_j)_{j \geq 0}$ with an infinite number of entries.  

\noindent For $X$ a subset of vertices of the DAG, $|X|_n$ denotes the
number of vertices of $X$ of rank $n$, and we define the \emph{width}
of $X$ as $\width(X) = \limsup_n |X|_n$. We use several times the
following property of the width.

$(\dagger)$ If $X_0 \neq \emptyset$ and $X_1$ are two disjoint
successor-closed sets, and if $X_0 \cup X_1 \subseteq X$, then then
$\width(X_1) < \width(X)$.

Let us prove property $(\dagger)$.  Let $r$ be the minimal rank of
vertices in $X_0$. Since $X_0$ is successor-closed and there is no
dead-end in the DAG, for every $n\geq r$, $X_0$ contains at least one
vertex of rank $n$. Because $X_0$ and $X_1$ are disjoint, we derive
$|X_1|_n + 1 \leq |X|_n$. Taking the limsup of this inequality we
obtain $(\dagger)$.

We pick $X$ a successor-closed set of nodes with infinitely many
incoming edges, of minimal width with this property. Let $v$ be a
vertex of $X$ of minimal rank and denote $S(v)$ for the set of
successors of $v$.  Let us show that $S(v)$ has infinitely many
incoming edges. Define $T(v)$ as the set of predecessors of vertices
in $S(v)$ and $Y = X \setminus T(v)$.  Then $Y$ is successor-closed
because $T(v)$ is predecessor-closed and $X$ is successor-closed.
Applying property $(\dagger)$ to $X_0=S(v)\subseteq X$ and
$X_1=Y\subseteq X$, we obtain $\width(Y) < \width(X)$.  By width
minimality of $X$ among successor-closed sets with infinitely many
incoming edges, $Y$ must have finitely many incoming edges only.
Since $Y = X \setminus T(v)$ and $X$ has infinitely many incoming
edges, then $T(v)$ has infinitely many incoming edges.  Thus there are
infinitely many edges connecting a vertex outside $S(v)$ to a vertex
of $S(v)$, so that $S(v)$ has infinitely many incoming edges.
\end{proof}
}

In this case, we say that index $i$ {\em leaks infinitely often}. 
Note that if $G$ separates $(r,t)$, and $r,
t$ have a common successor by $H$, then $G$ leaks at $H$.
To link leaks with separations, we consider for each index $k$, the pairs of states that have a common successor, in possibly several steps, as expressed by the symmetric relation $R_k$: $(r,t) \in R_k$ iff there exists $j\geq k$ and $y \in \states$ such that $(r,y) \in G[k,j]
\wedge (t,y) \in G[k,j]$.  
%
%
%

\begin{lemma}
\label{lemma.separation}
For $i<n$ two indices, the following three properties hold: 
\begin{enumerate}
\item If $G[i,n]$ separates $(r,t)\in R_{n}$, then there exists $m \geq n$ such that $G[i,m]$ leaks at $G_{m+1}$.

\item If index $i$ does not leak infinitely often, then the number of indices
  $j$ such that $G[i,j]$ separates some $(r,t)\in R_{j}$ is finite.

\item If index $i$ leaks infinitely often, then for all $j>i$,  $G[i,j]$ separates some $(r,t)\in R_{j}$.
\end{enumerate}
\end{lemma}
\arkiv{
\begin{proof}
  We start with the proof of the first item.  Assume that $G[i,n]$
  separates a pair $(r,t) \in R_n$.  Hence there exists $q$ such that
  $(q,r) \in G[i,n]$, $(q,t) \notin G[i,n]$.  Now, from $(r,t)\in
  R_n$, we derive the existence of an index $k > n$ and a state
  $y$ such that $(r,y) \in G[n,k]$ and $(t,y) \in G[n,k]$.  Hence,
  there exists a path $(t_j)_{n \leq j \leq k}$ with $t_n=t$, $t_{k}=y$,
  and $(t_j,t_{j{+}1}) \in G_{j+1}$ for all $n \leq j < k$.  Moreover,
  there is a path from $q$ to $y$ because there are paths from $q$ to
  $r$ and from $r$ to $y$.  Let $\ell \leq k$ be the minimum index
  such that there is a path from $q$ to $t_\ell$.  As there is no path
  from $q$ to $t_n=t$, necessarily $\ell \geq n+1$. Obviously,
  $(t_{\ell-1},t_{\ell}) \in G_{\ell}$, and by definition and
  minimality of $\ell$, $(q,t_{\ell-1}) \notin G[i,\ell-1]$ and
  $(q,t_{\ell}) \in G[i,\ell]$. That is, $G[i,\ell-1]$ leaks at
  $G_{\ell}$.

\medskip

Let us now prove the second item, using the first one. Assume that
$i$ does not leak infinitely often, and towards a contradiction
suppose that there are infinitely many $j$'s such that $G[i,j]$
separates some $(r,t) \in R_{j}$. To each of these separations, we
can apply item {\bf \emph{1.}} to obtain infinitely many indices
$m$ such that $G[i,m]$ leaks at $G_{m+1}$, a contradiction.

\medskip

We now prove the last item. Since there are finitely many states in
$\states$, there exists $\state \in \states$ and an infinite set $J$
of indices such that for every $j \in J$, $(q,y_{j{+}1}) \in G[i,j{+}1]$,
$(q,x_j) \notin G[i,j]$, and $(x_j,y_{j{+}1}) \in G_{j{+}1}$ for some
$x_j,y_{j{+}1}$. The path from $q$ to $y_{j{+}1}$ implies the existence of
$y_j$ with $(q,y_j) \in G[i,j]$, and $(y_j,y_{j{+}1}) \in G_{j{+}1}$. We
thus found separated pairs $(x_j,y_j)$ for every $j \in J$. To exhibit
separations at other indices $k >j$ with $k \notin J$, the natural
idea is to consider predecessors of the $x_j$'s and $y_j$'s.

\begin{center}
\begin{tikzpicture}
\node at (0,.75) {$i$};
\node at (3,.75) {$k$};
\node at (6,.75) {$j$};
\node at (8,.75) {$j{+}1$};

\draw (0,.5) -- (0,-1);
\draw (3,.5) -- (3,-1);
\draw (6,.5) -- (6,-1);
\draw (8,.5) -- (8,-1);

\draw (0,0) node (q) {$\bullet$};
\node  at (q.180) {$q$}; 

\draw (3,0) node (q') {$\bullet$};
\node  at (q'.35) {$r_k$}; 

\draw (3,-.5) node (tk) {$\bullet$};
\node  at (tk.325) {$t_k$}; 

\draw(6,-.5) node (x) {$\bullet$};
\node  at ([shift={(.5:.2)}]x.300) {$x_j$}; 
\draw(6,0) node (x') {$\bullet$};
\node  at ([shift={(-.5:.4)}]x'.110)  {$y_j$}; 
\draw(8,-.25) node (y) {$\bullet$};
\node  at ([shift={(.5:.5)}]y.0) {$y_{j{+}1}$}; 
\draw[-latex',thick] (q) -- (q');
\draw[-latex',thick] (q') -- (x');
\draw[-latex',thick] (x') -- (y);
\draw[-latex',thick] (x) -- (y);
\draw[-latex',thick] (tk) -- (x);

\node at (1.5,-1) {$G[i,k]$};
\node at (4.5,-1) {$G[k,j]$};
\node at (7,-1) {$G_{j+1}$};
\end{tikzpicture}
\end{center}

\noindent We define sequences $(r_k,t_k)_{k \geq i}$ inductively as follows.  To
define $r_k$, we take a $j \geq k+1$ such that $j \in J$; this is
always possible as $J$ is infinite.  There exists a state $r_k$
such that $(q,r_k) \in G[i,k]$ and $(r_k,y_j) \in G[k,j]$.

\noindent Also, as $x_j$ belongs to $\im(G[1,j])$, there must
exist a state $t_k$ such that $(t_k,x_j)\in G[k,j]$.  Clearly, $(q,
t_k) \notin G[i,k]$, else $(q,x_j) \in G[i,j]$, which is not
true.  Last, $y_{j{+}1}$ is a common successor of $t_k$ and
$r_k$, that is $(t_k,y_{j{+}1}) \in G[k,j+1]$ and
$(r_k,y_{j{+}1}) \in G[k,j+1]$. Hence $G[i,k]$ separates
$(r_{k},t_{k}) \in R_{k}$.
\end{proof}
}

\subsection{The tracking list}
\label{subsec:tl}
The \emph{tracking list} exploits the relationship between leaks and
separations. It is a list of transfer graphs which altogether separate
all possible pairs, and are sufficient to detect when leaks
occur. Notice that telling at step $j$ whether the pair $(r,t)$
belongs to $R_j$ cannot be performed by a deterministic automaton.
We thus \emph{a priori} have to
consider every pair $(r,t) \in \states^2$ for separation.  The
tracking list ${\cal L}_n$ at step $n$ is defined inductively
as follows.  ${\cal L}_0$ is the empty list, and for $n>0$, the list
${\cal L}_{n}$ is computed in three stages:
\begin{enumerate}
\item first, every graph $H$ in the list ${\cal L}_{n-1}$ is
  concatenated with $G_{n}$, yielding $H \cdot G_{n}$;
\item second, $G_{n}$ is added at the end of the obtained list;
\item last, the list is filtered: a graph $H$ is kept if and only if
  it separates a pair of states $(p,q)\in \states^2$ which is not separated by any graph that appears earlier in the list.
\end{enumerate}
Because of the third item, there are at most $|\states|^2$ graphs in
the tracking list. The list may become empty if no pair of states is
separated by any graph, for example if all the graphs are
complete. Let ${\cal L}_n = \{H_1 ,\cdots ,H_\ell\}$ be the tracking
list at step $n$. Then each transfer graph $H_r \in {\cal L}_n$ is of
the form $H_r = G[t_r,n]$. We say that $r$ is the \emph{level} of
$H_r$, and $t_r$ the \emph{index tracked} by $H_r$. Observe that the
lower the level of a graph in the list, the smaller the index it
tracks.
%
%
When we consider the sequence of tracking lists $({\cal L}_n)_{n \in
  \mathbb{N}}$, for every index $i$, either it eventually stops to be
tracked or it is tracked forever from step $i$,
\emph{i.e.}\ for every $n \geq i$, $G[i,n]$ belongs to ${\cal
  L}_n$. In the latter case, $i$ is said to be \emph{remanent} (because it will never disappear).

Using Lemma~\ref{lemma.leaks} and the second and third statements of
Lemma~\ref{lemma.separation}, we obtain:

\begin{lemma}
\label{lem:caracPG}
A play has infinite capacity iff there exists an index $i$ such that
$i$ is remanent and leaks infinitely often.
\end{lemma}
\arkiv{
\begin{proof}
  The direction from right-to-left is trivial. Assume the play has
  finite capacity, and let $i$ be a remanent index. By Lemma
  \ref{lemma.leaks}, $i$ does not leak infinitely oten.

\medskip

For the other direction, assume that the play has infinite capacity.
By Lemma~\ref{lemma.leaks}, there exists an index $i$ that leaks
infinitely often. We choose $i$ minimal with this property.

We first show that for all $k\geq i$, $k$ leaks infinitely often as well. There are infinitely many indices $j > i$ such that $G[i,j]$ leaks at $G_{j{+}1}$. For each such index $j$, there are states $q,x,y$ such that
$(q,y) \in G[i,j{+}1]$, $(q,x) \notin G[i,j]$ and $(x,y) \in
G_{j{+}1}$. Consider any index $i\leq k\leq j$. There exists a state
$q'$ such that $(q,q') \in G[i,k]$ and $(q',y) \in G[k,j{+}1]$. We thus
have $(q',y) \in G[k,j{+}1]$, $(q',x) \notin G[k,j]$ and $(x,y) \in
G_{j{+}1}$. Thus $G[k,j]$ leaks at $G_{j{+}1}$.  This holds for all $j >i$
and $i \leq k \leq j$, so that for all $k \geq i$, $G[k,j]$ leaks at
$G_{j{+}1}$ for infinitely many indices $j$.

\begin{center}
\begin{tikzpicture}
\node at (0,.75) {$i$};
\node at (3,.75) {$k$};
\node at (6,.75) {$j$};
\node at (8,.75) {$j{+}1$};

\draw (0,.5) -- (0,-1);
\draw (3,.5) -- (3,-1);
\draw (6,.5) -- (6,-1);
\draw (8,.5) -- (8,-1);

\draw (0,0) node (q) {$\bullet$};
\node  at (q.180) {$q$}; 

\draw (3,0) node (q') {$\bullet$};
\node  at (q'.325) {$q'$}; 

\draw(6,-.5) node (x) {$\bullet$};
\node  at (x.210) {$x$}; 
\draw(6,0) node (x') {$\bullet$};
\draw(8,-.25) node (y) {$\bullet$};
\node  at (y.0) {$y$}; 
\draw[-latex',thick] (q) -- (q');
\draw[-latex',thick] (q') -- (x');
\draw[-latex',thick] (x') -- (y);
\draw[-latex',thick] (x) -- (y);

\node at (1.5,-1) {$G[i,k]$};
\node at (4.5,-1) {$G[k,j]$};
\node at (7,-1) {$G_{j+1}$};
\end{tikzpicture}
\end{center}

We prove now that some $k \geq i$ is remanent, which will finish the proof. Towards a contradiction, assume that it is not the case.

Let $\ell < i$. By minimality of $i$, $\ell$ leaks only finitely
often. Applying the second statement of Lemma~\ref{lemma.separation},
there are only finitely many indices $j \geq \ell$ such that
$G[\ell,j]$ separates some pair of $R_{j}$. We let
$j_\ell$ the maximum of these indices, and $N = \max_{\ell <i}
j_\ell$. By definition of $N$, for all $\ell <i$ and all $j >N$,
$G[\ell,j]$ separates no pair of $R_{j}$.

Fix now $n>N$, the minimal index such that there exists $j$ with $G[n,j] \in
{\cal L}_j$ and for all $i \leq k \leq N$, $G[k,j] \notin {\cal
  L}_j$. The existence of $n$ is guaranteed since we assumed for
contradiction that no $k \geq i$ is remanent.  Let $J$ be the step at
which index $n$ is no longer tracked in the list. Just before the list
is filtered to obtain ${\cal L}_J$, it starts with a prefix of the form: $G[i_1,J], \cdots, G[i_\ell,J], G[n,J]$. By definition of $n$, the indices $i_1, \cdots, i_\ell$ are smaller than $i$. That is, $i_1 <
\cdots <i_\ell < i < N < n \leq J$.

Now, the choice of $N$ guarantees that for all $1 \leq k \leq \ell$,
$G[i_k,J]$ separates no pair in $R_{J}$. Moreover, $n \geq i$ thus $n$
leaks infinitely often, and by the third statement of
Lemma~\ref{lemma.separation}, $G[n,J]$ separates some pair of $R_{J}$,
which cannot be separated by any $G[i_k,J]$. Therefore, during the
third stage, $G[n,J]$ is not filtered. This contradicts the definition
of $J$ as the step after which index $n$ is no longer tracked.

Thus some index larger than $i$ is remanent, and leaks infinitely often.
\end{proof}
}

\subsection{The parity game}
We now describe a parity game $\PG$, which extends the support arena with 
on-the-fly computation of the tracking list.

{\medskip \noindent \bf Priorities.} By convention, lowest priorities are the most important
and the odd parity is good for \playerone, so \playerone\ wins iff the $\liminf$ of the priorities
is odd. With each level $1\leq r \leq |\stnb|^2$ of the tracking list are associated two priorities
$2r$ and $2r+1$, and on top of that are added priorities $1$ and $2 |\stnb|^2+2$,
hence the set of all priorities is $\{1,\ldots,2 |\stnb|^2+2\}$. 

When \playertwo\ chooses a transition labelled by a transfer graph
$G$, the tracking list is updated with $G$ and the priority of the
transition is determined as the smallest among: priority 1 if the
support  $\{f\}$ has ever been visited, priority $2r+1$ for the
smallest $r$ such that $H_r$ (from level $r$) leaks at $G$,
priority $2r$ for the smallest level $r$ where a graph was removed,
and in all other cases priority $2 |\stnb|^2+2$.
%
%

{\medskip \noindent \bf  States and transitions.}
$\GG^{\leq |\stnb|^2}$ denotes the set of list of at most $|\stnb|^2$ transfer graphs.
\begin{itemize}
\item 
States of $\PG$ form a subset of $\{0,1\}  \times 2^Q\times \GG^{\leq |\stnb|^2}$,
each state being of the form
$(b,S,H_1,\ldots,H_\ell)$ with $b \in \{0,1\}$ a bit indicating 
whether a support in $\{f\}$ has been seen, $S$ the current support
and $(H_1,\ldots,H_\ell)$ the tracking list.
%
The initial state is $(0,\{\state_0\},\emptyset)$.

\item Transitions in $\PG$ are all $(b,S,H_1,\ldots,H_{\ell})
  \arr{\priority,a,G} (b',S',H'_1,\ldots,H'_{\ell'})$ where $\priority$
  is the priority, and such that $S \arr{a,G} S'$ is a transition of
  the support arena, and
  \begin{enumerate}
  \item $(H'_1,\ldots,H'_{\ell'})$ is the tracking list obtained by
    updating the tracking list $(H_1,\ldots,H_{\ell})$ with $G$, as
    explained in subsection~\ref{subsec:tl};
    \item if $b=1$ or if $S' \subseteq \targetset$, then 
     $\priority =1$ and $b'=1$;
\item otherwise $b'=0$.
In order to compute the priority $\priority$, we let $\priority'$ be the smallest level $1\leq r \leq \ell$ such that $H_r$ 
leaks at $G$ and $\priority'=\ell+1$ if there is no such level, and we
also let
$\priority''$ as the minimal level $1\leq r \leq \ell$ such that $H'_{r}
\neq H_{r} \cdot G$ and $\priority''=\ell+1$ if there is no such
level. Then $\priority=\min( 2 \priority'+1, 2 
\priority'')$.  
\end{enumerate}
\end{itemize}


We are ready to state the main result of this paper, which yields an
\EXPTIME complexity for the population control problem. This entails
the first statement of Theorem~\ref{th1}, and together with
Proposition~\ref{th.2wins}, also the first statement of Theorem~\ref{th2}.

\begin{theorem}
  \playerone\ wins the game $\PG$ if and only if \playerone\ wins the
  \ror\ game.
  Solving these games can be done in time
  $O(2^{(1+|\stnb|+|\stnb|^4)(2|\stnb|^2+2)})$. 
Strategies with $2^{|\stnb|^4}$ 
memory states are sufficient to both \playerone\ and \playertwo.
\end{theorem}

\begin{proof}
{\color{black} 
The state space of parity game $\PG$ is the product of the 
set of supports  
with a deterministic automaton computing the tracking list.
As the state space of the \ror\ game is also 
the set of supports, there is a natural correspondence between plays and strategies in the parity game $\PG$
and in the \ror\ game.

\playerone\ can win the parity game $\PG$ in two ways: either the play
visits the support $\{f\}$, or the priority of the play is $2r+1$ for
some level $1\leq r \leq |Q|^2$. By design of $\PG$, this second
possibility occurs iff $r$ is remanent and leaks infinitely often.
According to Lemma~\ref{lem:caracPG}, this occurs if and only if the
corresponding play of the \ror\ game has infinite capacity.  Thus
\playerone\ wins $\PG$ iff he wins the capacity game.

In the parity game $\PG$, there are at most $2^{1 +
  |\stnb|}\left(2^{|\stnb|^2}\right)^{|Q|^2}=2^{1 + |\stnb| +
  |\stnb|^4}$ states and $2|Q|^2+2$ priorities, implying the
complexity bound using state-of-the-art
algorithms~\cite{JurdzinskiStacs2000}.  
Actually the complexity is even quasi-polynomial according to the
algorithms in~\cite{sanjay}.  Notice however that this has little impact on the complexity of the population control problem, as the number of priorities is logarithmic in the number of states of our parity game.

Further, it is well known that the winner of a parity game has a
positional winning strategy~\cite{JurdzinskiStacs2000}. 
A \emph{positional} winning strategy $\sigma$ in the game $\PG$
corresponds to a \emph{finite-memory} winning strategy $\sigma'$ in
the \ror\ game, whose memory states are the states of $\PG$. Actually
in order to play $\sigma'$, it is enough to remember the tracking list,
\emph{i.e.}\ the third component of the state space of $\PG$. Indeed,
the second component, in $2^Q$, is redundant with the actual state of
the \ror\ game and the bit in the first component is set to $1$ when
the play visits $\{f\}$ but in this case the \ror\ game is won by
\playerone\ whatever is played afterwards. Since there at most
$2^{|Q|^4}$ different tracking lists, we get the upper bound on the
memory.} { }
\end{proof}
	\section{Lower bounds}

The proofs of Theorems~\ref{th1} and~\ref{th2} are concluded by 
the proofs of lower bounds.

\begin{theorem}
The population control problem is \EXPTIME-hard.
\end{theorem}

\begin{proof}
  We first prove \PSPACE-hardness of the population control problem,
  reducing from the halting problem for polynomial space Turing
  machines.  We then extend the result to obtain the
  \EXPTIME-hardness, by reducing from the halting problem for
  polynomial space {\em alternating} Turing machines.  Let
  $\mathcal{M} = (S,\Gamma,T,s_0,s_f)$ be a Turing machine with
  $\Gamma = \{0,1\}$ as tape alphabet. By assumption, there exists a
  polynomial $P$ such that, on initial configuration
  $x \in \{0,1\}^n$, $\mathcal{M}$ uses at most $P(n)$
  tape cells. A transition $t \in T$ is of the form
  $t=(s, s', b, b', d)$, where $s$ and $s'$ are, respectively, the
  source and the target control states, $b$ and $b'$ are,
  respectively, the symbols read from and written on the tape, and
  $d \in \{\leftarrow,\rightarrow,-\}$ indicates the move of the tape
  head. From $\mathcal{M}$ and $x$, we build an NFA
  $\nfa = (\states, \Sigma, \state_0,\Delta)$ with a distinguished
  state $\smiley$ such that, $\mathcal{M}$ terminates in $s_f$ on
  input $x$ if and only if $(\nfa,\smiley)$ is a positive instance of
  the population control problem.  

  The high-level description of $\nfa$ is as follows. States in
  $\states$ are of several types: contents of the $P(n)$ cells (one
  state $(b,p)$ per content and per position), position of the tape
  head (one state $p$ per possible position), control state of the
  Turing machine (one state $s$ per control state), and three special
  states, namely an initial state $q_0$, a sink winning state
  $\smiley$, and a sink losing state $\frownie$. With each transition
  $t=(s, s', b, b', d)$ in the Turing machine and each position $p$ of
  the tape, we associate an action $a_{t,p}$ in $\nfa$, which
  simulates the effect of transition $t$ when the head position is
  $p$. Thus, on action $a_{t,p}$ there is a transition from the source
  state $s$ to the target state $s'$, another from the tape head
  position $p$ to its update according to $d$, and also from $(b,p)$
  to $(b',p)$. Moreover, from head position $q \neq p$, $a_{t,p}$
  leads to $\frownie$, so that in any population game, \playerone\
  only plays actions associated with the current head
  position. Similarly from states $(b'',p)$ with $b'' \neq b$, states
  $s'' \neq s$, action $a_{t,p}$ leads to $\frownie$.  Initially, an
  $\init$ action is available from $\state_0$ and leads to $s_0$, to
  position $0$ for the tape head, and to cells $(b,p)$ that encode the
  initial tape contents on input $x$. The NFA also has winning
  actions, that allow one to check that there are no agents in a
  subset of states, and send the remaining ones to the target
  $\smiley$. One such action should be played when agents encoding the
  state of the Turing machine lie in $s_f$, indicating that
  $\mathcal{M}$ accepted.  Another winning action $\winact$ is played
  whenever there are not enough agents to encode the initial
  configuration: \playertwo\ needs $m$ to be at least $P(n)+2$ to fill
  states corresponding to the initial tape contents ($P(n)$ tokens),
  the initial control state $s_0$ and the initial head position.  The
  sink losing state $\frownie$ is used to pinpoint an error in the
  simulation of $\mathcal{M}$.

\medskip

Now, in order to encode an alternating Turing machine, we assume that
the control states of $\mathcal{M}$ alternate between states of \playerone\
and states of \playertwo. The NFA $\nfa$ is extended with a
state $C$, for \playerone,  and an additional transition labelled
$\init$ from $q_0$ to $C$.  Assume first, that $C$ contains at
most an agent; we will later explain how to impose this.
Beyond $C$, the NFA also contains on state $t$ per transition 
of $\mathcal{M}$, which will represent that \playertwo\ chooses transition $t$.  
To do so, from state $C$, for any action $a_{t,p}$, there are
transitions to all states $t'$. 
From state $t$, actions of the form $a_{t,p}$ are allowed, leading
back to $C$. That is, actions $a_{t',p}$ with $t' \neq t$ lead from
$t$ to the sink losing state $\frownie$. This encodes that \playerone\
must follow the transition $t$ chosen by \playertwo. To punish
\playertwo\ in case the current tape contents is not the one expected by
the transition $t=(s, s', b, b', d)$ he chooses, there are trashing actions
$\trash_{s}$ and $\trash_{p,b}$ enabled from state $t$.  Action
$\trash_{s}$ leads from $t$ to $\smiley$, and also from $s$ to
$\frownie$.  In this way, the agents in $t$ cannot be used by
\playertwo, and \playerone\ wins more easily.  Similarly,
$\trash_{p,b}$ leads from $t$ to $\smiley$ and from any position state
$q \neq p$ to $\frownie$, and from $(b,p)$ to $\frownie$. 

Last, there are transitions on action $\ttend$ from state $\smiley$, $C$
and any of the $t$'s to the target state $\smiley$. Moreover, action
$\ttend$ from any other state (in particular the ones encoding the Turing
machine configuration) to $\frownie$.  This whole construction
encodes, assuming that there is a single agent in $C$ after the first
transition, that \playerone\ can choose the transition from a
\playerone\ state of $\mathcal{M}$, and \playertwo\ can choose
the transition from an \playertwo\ state.

\medskip

Let us now explain the gadget, represented below, to deal with the
case where \playertwo\ places several agents in state $C$ on the
initial action $\init$, enabling the possibility to later send agents
to several $t$'s simultaneously.
We use an extra 
state $\sstore$, actions $\store_t$ for each transition $t$, and action
$\restart$.
Action
$\store_{t}$ is a left loop on every state except from $t$, which goes
to $\sstore$. 
From all states except $\smiley$ and $\frownie$ action $\restart$
leads to $q_0$.  Last, the effects of $\winact$ and $\ttend$ are
modified as follow: $\winact$ leads from $C$ and from any $t$ to
$\smiley$, it loops on $s_f$ and moves from all other $s$'s to
$\frownie$; $\ttend$ goes from all $s$'s and $\smiley$ to $\smiley$,
and transitions from $q_0$, $C$, the $t$'s and $\sstore$ to
$\frownie$.

\vspace{-0.4cm}
\begin{figure}[h!]
{\center
\begin{center}
\includegraphics[scale=0.3]{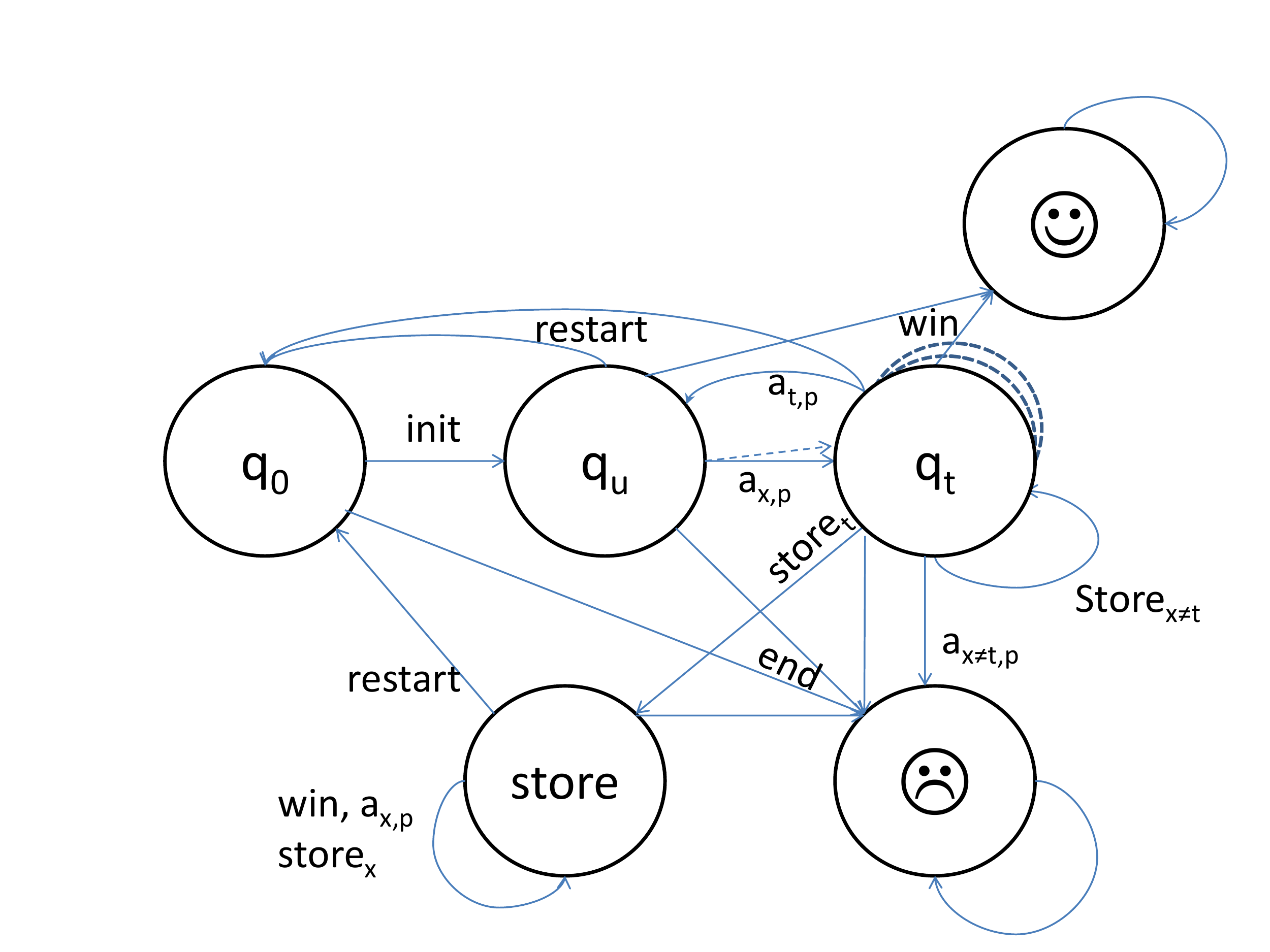}
\end{center}
}
\end{figure}

Assume that input $x$ is not accepted by the alternating Turing
machine $\mathcal{M}$, and let $m$ be at least $P(n)+3$.
In the $m$-population game, \playertwo\ has a winning strategy placing
initially a single agent in state $C$. If \playerone\ plays
$\store_{t}$ (for some $t$), either no agents are stored, or the
unique agent in $C$ is moved to $\sstore$. Thus \playerone\ cannot
play $\ttend$ and has no way to lead the agents encoding the Turing
machine configuration to $\smiley$, until he plays $\restart$, which
moves all the agents back to $q_0$. This shows that $\store_{t}$ is
useless to \playerone\, and thus \playertwo\ wins.

Conversely, if \playerone\ has a strategy in $\mathcal{M}$ witnessing
the acceptance of $x$, 
in order to win the $m$-population game, \playertwo\ would need to
cheat in the simulation of $\mathcal{M}$ and place at least two agents
in $C$ to eventually split them to ${t_1}, \ldots, {t_n}$. Then,
\playerone\ can play the corresponding actions
$\store_{t_2}, \ldots, \store_{t_n}$ moving all agents (but the ones
in ${t_1}$) in $\sstore$, after which he plays his winning strategy
from ${t_1}$ resulting in sending some agents to $\smiley$. Then,
\playerone\ plays $\restart$ and proceeds inductively with strictly
less agents from $q_0$, and eventually plays $\ttend$ to win.
\end{proof}

Surprisingly, the cut-off can be as high as doubly exponential in the size of the NFA.

\begin{proposition}
\label{prop:cutoff-lowerbound}
  There exists a family of NFA $(\nfa_n)_{n \in \nats}$ such that
  $|\nfa_n|=2n+7$, 
  and for $M = 2^{2^{n}+1}+n$, there is no winning strategy
  in $\nfa_n^{M}$ and there is one in $\nfa_n^{M-1}$.
\end{proposition}

{\begin{proof}
  Let $n\in \nats$. The NFA $\nfa_n$ we build is 
  the disjoin union of {\em two} NFAs with different properties, namely
 $\splitnfa,\countnfa{n}$.
 On the one hand, for $\splitnfa$, winning the game
  with $m$ agents requires $\Theta(\log m)$ steps. On the other
  hand, $\countnfa{n}$ implements a
  usual counter over $n$ bits
  (as used in many different publications),
  such that \playerone\ can avoid to lose during $O(2^n)$ steps. The combination of these two gadgets ensures a cut-off for $\nfa_n$ of $2^{2^n}$.

 Recall Figure~\ref{fig:splitgadget}, which presents the splitting gadget that has the following properties. 
In $\splitnfa^m$ with $m \in \nats$ agents,
$(s1)$ there is a winning strategy ensuring to win in $2 \left \lfloor \log_2 m \right \rfloor + 2$ steps;
$(s2)$ no strategy can ensure to win in less than $2 \left \lfloor \log_2 m \right \rfloor + 1$ steps.


{\blaise
The counting gadget that implements a counter with states $l_i$ (meaning bit $i$ is $0$) and $h_i$ (for bit $i$ is $1$) enjoys the following properties:}
$(c1)$ there is a strategy in
$\countnfa{n}$ to ensure avoiding $\frownie$ during $2^n$ steps,
by playing $\alpha_i$ whenever the counter suffix from bit $i$ is $0 1\cdots 1$;
$(c2)$ for $m \geq n$, no strategy of $\countnfa{n}^m$ avoid $\frownie$ for $2^n$ steps.

{\blaise
The two gadgets (splitting and counting) are combined
by a new initial state leading by two transitions labeled {\em init} to the initial states of both NFAs.
Actions consist of pairs of actions, one for each gadget: $\Sigma =
\{a,b,\delta\} \times \{\alpha_i \mid 1 \leq i \leq n\}$.}
We add an action $*$ which can be played from any state of $\countnfa{n}$ but $\frownie$, and only from $f$ in $\splitnfa$, 
leading to the global target state $\smiley$.

Let $M=2^{2^{n}+1}+n$. 
We deduce that the cut-off is $M-1$ as follows:
\begin{itemize}
\item 
For $M$ agents, a winning strategy for \playertwo\ is to
first split $n$ tokens from the initial state to the $q_0$ 
of $\countnfa{n}$, in order to fill each $l_i$ with 1 token, and
$2^{2^{n}+1}$ tokens to the $q_0$ of $\splitnfa$.
Then \playertwo\ splits evenly tokens between $q_1,q_2$ in $\splitnfa$.
In this way, \playerone\ needs at least $2^n+1$ steps
to reach the final state of $\splitnfa$ ($s2$), but \playerone\ reachs $\frownie$ after these $2^n+1$ steps in $\countnfa{n}$ ($c2$).

\item
For $M-1$ agents, \playertwo\ needs to use at least $n$ tokens 
from the initial state to the $q_0$ of $\countnfa{n}$, 
else \playerone\ can win easily. But then there are less than 
$2^{2^{n}+1}$ tokens in the $q_0$ of $\splitnfa$.
And thus by $(s1)$, \playerone\ can reach $f$ within $2^n$ steps,
after which he still avoids $\frownie$ in $\countnfa{n}$ ($c1$).
And then \playerone\ sends all agents to $\smiley$ using $*$.
\end{itemize}
Thus, the family $(\nfa_n)$ of NFA exhibits a doubly exponential cut-off.
\end{proof}
}

\section{Discussion}
\label{sec:conclu}
Obtaining an \EXPTIME algorithm for the control problem of a population of agents was challenging. We have a matching 
\EXPTIME-hard lower-bound. Further, the surprising doubly exponential matching upper and lower bounds on the cut-off imply that the alternative technique, checking that \playerone\ wins all $m$-population game for $m$ up to the cut-off, is far from being efficient.
This compares favourably  with the exponential gap for the parameterized verification of almost-sure reachability for a population communicating via a shared register~\cite{BMRSS-icalp16} (the latter problem is in \EXPSPACE and \PSPACE-hard). 

The idealised formalism we describe in this paper is not entirely satisfactory: for instance, while each agent can move in a non-deterministic way, unrealistic behaviours can happen, \emph{e.g.}\ all agents synchronously taking infinitely often the same choice.  An almost-sure control problem in a probabilistic formalism should be studied, ruling out such extreme behaviours.
As the population is discrete, we may avoid the undecidability that holds for distributions~\cite{DMS12} and is inherited from the equivalence with probabilistic automata~\cite{GO10}. Abstracting continuous distributions by a discrete population of arbitrary size could thus be seen as an approximation technique for undecidable formalisms such as probabilistic automata.

\medskip

{\bf Acknowledgement:} We are grateful to Gregory Batt for fruitful discussions concerning the biological setting. Thanks to Mahsa Shirmohammadi for interesting discussions. This work was partially supported by ANR projet STOCH-MC (ANR-13-BS02-0011-01).

\bibliographystyle{plainurl}
\bibliography{biblio}	
		
\end{document}